\documentclass{article}

\usepackage{textcomp,mathrsfs}
\usepackage{graphicx,amssymb,amsmath,amsthm,mathrsfs,tipa}
\usepackage{multirow,makeidx,algorithmic,algorithm}
\usepackage{latexsym,graphicx,amssymb,amsmath,amsthm,mathrsfs,tipa}
\usepackage{mathrsfs}
\usepackage{amssymb}
\usepackage{booktabs,fancyhdr} 
\usepackage[usenames,dvipsnames]{color}
\usepackage{url}
\begingroup
\makeatletter
\g@addto@macro{\UrlSpecials}{%
  \endlinechar=13 \catcode\endlinechar=12
  \do\%{\Url@percent}\do\^^M{\break}}
 \gdef\Url@percent{\@ifnextchar^^M{\@gobble}{\mathbin{\mathchar`\%}}}%
\endgroup %

\graphicspath{{./graphics/},{./figures/}}

\newcounter{xxx}
\def\vertexset{{\mathscr V}}
\def\edgeset{{\mathscr E}}
\def\leavesset{{\mathscr L}}
\def\observations{{\mathscr Y}}
\def\potentials{{\mathscr W}}

\newcommand\bracearraycond[1]{\left\{ \begin{array}{ll} #1 \end{array} \right.}

\def\N{{\mathbb N}}        % naturals
\def\M{{\mathbb M}}        % naturals
\def\R{{\mathbb R}}        % reals
\def\P{{\mathbb P}}        % probability
\def\1{{\mathbf 1}}        % indicator
\DeclareMathOperator{\parent}{pa}

\newtheorem{theorem}{Theorem}
\newtheorem{lemma}[theorem]{Lemma}
\newtheorem{proposition}[theorem]{Proposition}

\newtheorem{definition}[theorem]{Definition}

\newcommand\indexedmtx{{M}} %{{\mathcal M}}
\newcommand\mindex{{(\textnormal{m})}}
\newcommand\pindex{{(\textnormal{p})}}
\newcommand\efindex{{(\textnormal{f})}}
\newcommand\findex{{(1)}}
\newcommand\sindex{{(2)}}
\newcommand\iindex{{(i)}}

\newcommand\xspace{{\mathcal X}}
\newcommand\transspace{{\mathcal Q}}

\usepackage[top=1.2in, bottom=1.2in, left=1.2in, right=1.2in]{geometry}

\usepackage{graphicx}

\title{A Note on Probabilistic Models over Strings: the Linear Algebra Approach}
\author{Alexandre Bouchard-C\^ot\'e}
 
\begin{document}

\maketitle

\begin{abstract}
Probabilistic models over strings have played a key role in developing methods allowing indels to be treated as phylogenetically informative events. 
There is an extensive literature on using automata and transducers on phylogenies to do inference on these probabilistic models, in which an important theoretical question in the field is the complexity of computing the normalization of a class of string-valued graphical models. This question has been investigated using tools from combinatorics, dynamic programming, and graph theory, and has practical applications in Bayesian phylogenetics.
In this work, we revisit this theoretical question from a different point of view, based on linear algebra. 
The main contribution is a new proof of a known result on the complexity of inference on TKF91, a well-known probabilistic model over strings. Our proof uses a  different approach based on classical linear algebra results, and is in some cases easier to extend to other models. The proving method also has consequences on the implementation and complexity of inference algorithms.
\end{abstract}

{\bf keywords} indel, alignment, probabilistic models, TKF91, string transducers, automata, graphical models, phylogenetics, factor graphs

% \newpage 

\section{Introduction}\label{sec:introduction}

This work is motivated by models of evolution that go beyond substitutions by incorporating insertions and deletions as phylogenetically informative events \cite{Thorne1991a,Thorne1992,Miklos2001,miklos2004,Bouchard2012Evolutionary}. In  recent years, 
there has been significant progress in  
turning these models into phylogenetic reconstruction methods that do not assume that the sequences have been aligned as a pre-processing step, and into statistical alignment methods
\cite{Hein1990Unified,Metzler2001,Miklos2003,Knudsen2003LongApprox,Redelings2005,Rivas2005,Redelings2007Baliphy2,Bouchard2012PIP}. Joint phylogenetic reconstruction and alignment methods are appealing because they can produce calibrated confidence assessments \cite{Rivas2005}, they leverage more information than pure substitution models \cite{Kawakita2003IndelInformative}, and they are generally less susceptible to biases induced by conditioning on a single alignment \cite{Wong2008}.

Model-based joint phylogenetic methods require scoring each tree in a large collection of hypothesized trees \cite{Redelings2005}. This collection of proposed trees is generated by a Markov chain Monte Carlo (MCMC) algorithm in Bayesian methods, or by a search algorithm in frequentist methods. Scoring one hypothesized tree then boils down to computing the probability of the observed sequences given the tree.
This computation requires summing over all possible ancestral sequences and derivations (a refinement of the notion of alignment, localizing all the insertions, deletions and substitutions on the fixed tree that can yield the observed sequences) \cite{Holmes2001}. 
Since the length of the ancestral strings is a priori unknown, these summations range over a countably infinite space, creating a challenging problem. This situation is in sharp contrast to the classical setup where an alignment is fixed, in which case computing the probability of the observed aligned sequences can be done efficiently using the standard Felsenstein recursion \cite{Felsenstein1981}.
 
% There is an extensive literature on using automata and transducers on trees to attack this problem, which in turn builds on combinatorics, dynamic programming, and graph theory \cite{Hein2001Generalization,Holmes2001,Steel2001TKF,Hein2001TKF,Lunter2003TKF,Hein2003RecursionsPNAS,Dreyer2008,Bouchard-Cote2009b,Mohri2009}. 
% Our work is most closely related to \cite{Westesson2012Automata,Westesson2012Plos}, 

When an alignment is not fixed, the distribution over a single branch is typically expressed as a transducer \cite{Hein1990Unified}, and the distribution over a tree is obtained by composite transducers and automata \cite{Holmes2001}.
The idea of composite transducers and automata, and their use in phylogenetic, have been introduced and subsequently developed in previous work in the framework of graph representations of transducers and combinatorial methods \cite{Holmes2001,Hein2001Generalization,Hein2001TKF,Hein2003RecursionsPNAS,Lunter2005,Dreyer2008}. Our contribution lies in new proofs, and a theoretical approach purely based on linear algebra. We use this linear algebra formulation to obtain simple algebraic expressions for marginalization algorithms. In the general cases, our representation provides a slight improvement on the asymptotic worst-case running time of existing transducer algorithms \cite{Mohri2009}. We also characterize a subclass of problems, which we call \emph{triangular problems},  where the running time can be further improved.  We show, for example, that marginalization problems derived from the popular TKF91 \cite{Thorne1991a} are triangular. We use this to develop a new complexity analysis of the problem of inference in TKF91 models. The bounds for TKF91 on star trees and perfect trees coincide with those obtained from previous work \cite{Lunter2003TKF}, but our proving method is easier to extend to other models since it only relies on the triangularity assumption rather than on the details of the TKF91 model.

Algorithms for composite phylogenetic automata and transducers have been introduced in \cite{Holmes2001,Holmes2003GuideTree}, generalizing the dynamic programming algorithm of \cite{Hein2001Generalization} to arbitrary guide trees. 
Other related work based on combinatorics includes \cite{Steel2001TKF,Lunter2003TKF,Song2006Recurs,Satija2008Combining,Bouchard-Cote2009b}.  See \cite{Bradley2007} for an outline of the area, and \cite{Holmes2007Phylocomposer,Redelings2005,Novak2008StatAlign} for computer implementations. 
Our work is most closely related to \cite{Westesson2012Automata,Westesson2012Plos}, which extends Felsenstein's algorithm to transducers and automata. They use a different view based on  partial-order graphs to represent ensemble profiles of ancestral sequences.

While the exact methods described here and in previous work are exponential in the number of taxa, they can form the basis of efficient approximation methods, for example via existing Gibbs sampling methods \cite{Jensen2002,Redelings2005}. These methods augment the state space of the MCMC algorithm, for example  with internal sequences. The running time then depends on the number of taxa involved in the local tree perturbation. This number can be as small as two, but there is a trade-off between the mixing rate of the MCMC chain and the computational cost of each move.

Note that the matrices used in this work should not be confused with the rate matrices and their marginal transition probabilities. The latter are used in classical phylogenetic reconstruction setups where the alignment is assumed to be fixed.
Our proving methods are more closely related to  previous work on representation of stochastic automata network \cite{Fernandez1998SAN,Langville2004SAN}, where interactions represent synchronization steps rather than string evolution, and therefore behave differently.

\section{Formal description of the problem}\label{sec:background}

\subsection{Notation}\label{sec:inference-questions}

We  use $\tau$ to denote a rooted phylogeny, with topology $(\vertexset, \edgeset)$.\footnote{In reversible models such as TKF91, an arbitrary root can be selected without changing the likelihood.  In this case, the root is used for computational convenience and has no special phylogenetic meaning.}  The root of $\tau$ is denoted by $\Omega$; the parent of $v\in\vertexset, v\neq\Omega$, by $\parent(v)\in\vertexset$; the branch lengths by $b(v) = b(\parent(v)\to v)$; and the set of leaves by $\leavesset\subset \vertexset$.  

The models we are interested in are 
tree-shaped graphical models \cite{Jordan2004GraphicalModels,Bishop2006Book,Airoldi2007GraphicalModels} where nodes are string-valued random variables $X_v$ over a finite alphabet $\Sigma$ (for example a set of nucleotides or amino acids).  There is one such random variable for each node in the topology of the phylogeny, $v\in\vertexset$, and the variable $X_v$ can be interpreted as an hypothesized biological sequence at one point in the phylogeny.  

Each edge $(\parent(v) \to v)$ 
denotes evolution from one species to another, and is associated with a conditional probability $\P_\tau(X_v = s'|X_{\parent(v)} = s) = \theta_{v}(s,s')$, where $s,s'\in\Sigma^*$ are strings (finite sequences of characters from $\Sigma$).  These conditional probabilities can be derived from the marginals of continuous time stochastic process \cite{Thorne1991a,Holmes2001,Miklos2001,miklos2004}, or from a parametric family \cite{Knudsen2003LongApprox,Rivas2005,Redelings2005,Redelings2007Baliphy2}.  We also denote the distribution at the root by $\P_\tau(X_\Omega = s) = \theta(s)$, which is generally set to the stationary distribution in reversible models or to some arbitrary initial distribution otherwise.
Usually, only
the sequences at the terminal nodes 
are observed, an event that we denote by $\observations = (X_v = x_v : v \in \leavesset)$.

We denote the probability of the data given a tree by $\P_\tau(\observations)$.  The computational bottleneck of model-based methods is to repeatedly compute $\P_\tau(\observations)$ for different hypothesized phylogenies $\tau$.  This is the case not only in  maximum likelihood estimators, but also in Bayesian methods, where a ratio of the form $\P_{\tau'}(\observations)/\P_\tau(\observations)$ is the main contribution to the Metropolis-Hastings ratio driving  Markov Chain Monte Carlo (MCMC) approximations.

While it is more intuitive to describe a phylogenetic tree using a Bayesian network, it is useful when developing inference algorithms to transform these Bayesian networks into equivalent factor graphs \cite{Bishop2006Book}.
A factor graph encodes a weight function over a space $\xspace$ assumed to be countable in this work. For example, this weight function could be a probability distribution, but it is convenient to drop the requirement that it sums to one.   In phylogenetics, the space $\xspace$ is a tuple of $|\vertexset|$ strings, $\xspace = (\Sigma^*)^{|\vertexset|}$.   

A factor graph (see Figure~\ref{fig:factor-graph} for an example) is based on a specific bipartite graph, built by letting the first bipartite component correspond to $\vertexset$ (the nodes denoted by circles in the figure), and the second bipartite component, to a set of \emph{potentials} or \emph{factors} $\potentials$ described shortly (the nodes denoted by squares in the figure). 

We put an edge in the factor graph between $w\in\potentials$ and $v\in\vertexset$ if the value taken by $w$ depends on the string at node $v$.
Formally, a potential $w\in\potentials$ is a map taking elements of the product space, $(x_v : v \in \vertexset) \in \xspace$, and returning a non-negative real number.  However the value of the factors we consider only depend on a small subset of the variables at a time. More precisely, we will use two types of factors: unary factors, which depend on a single variable, and binary factors, which depend on two variables.  We use the abbreviation SFG for such string-valued factor graph.

\begin{figure}[tp]
\begin{center} 
\includegraphics[width=2.3in]{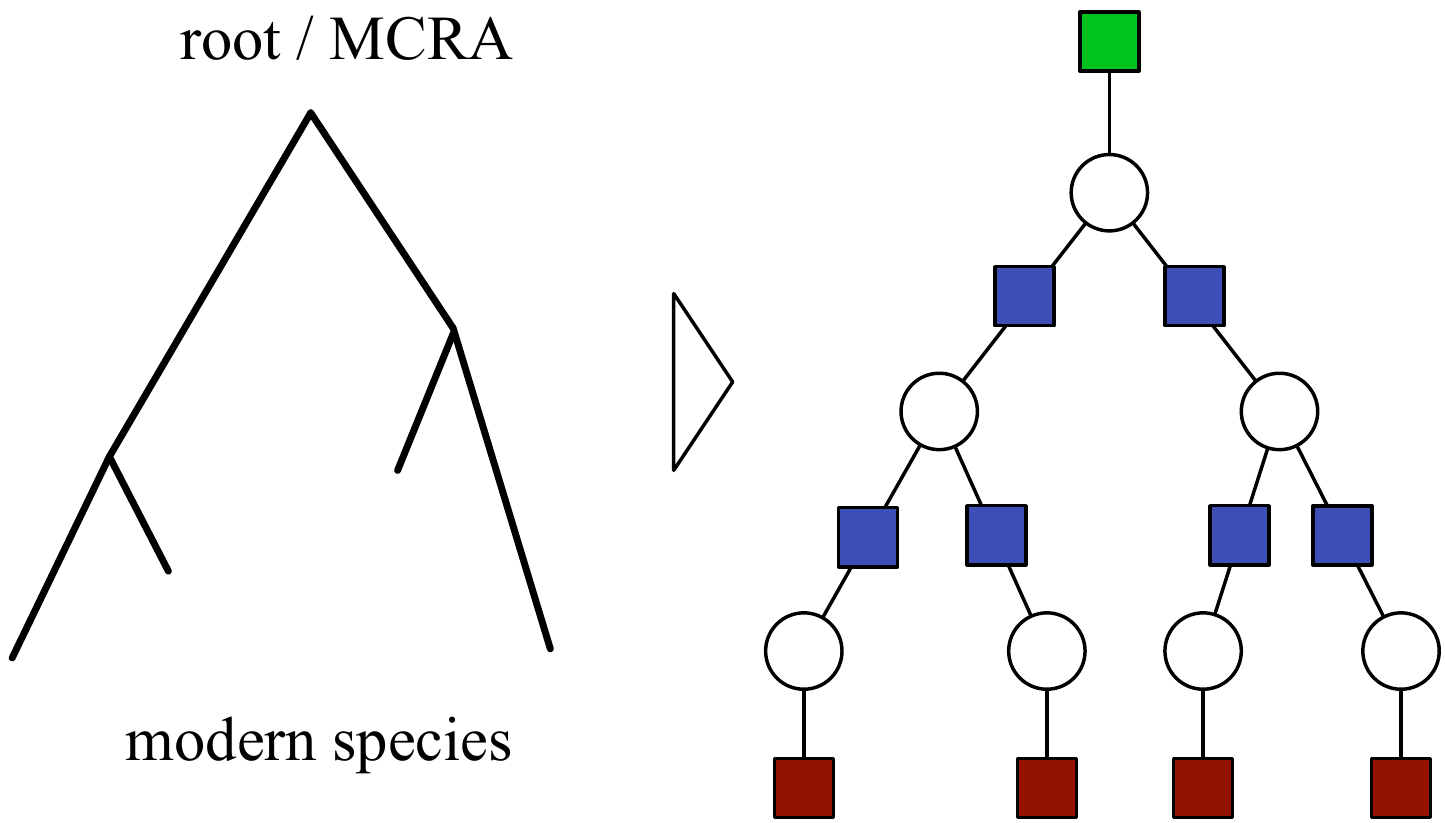}
  \caption[Factor graph construction]{Factor graph construction (right) derived from a phylogenetic tree (left): in green, the unary factor at the root capturing the initial string distribution; in blue, the binary factors capturing string mutation probabilities; and in red, the unary factor at the leaves, which are indicator functions on the observed strings.} 
  \label{fig:factor-graph}
\end{center} 
\end{figure}

In order to write a precise expression for the probability of the data given a tree, $\P_\tau(\observations)$,  we construct a factor graph (shown in Figure~\ref{fig:factor-graph}) with three types of factors:  one unary potential $w_\Omega(s) = \theta(s)$ at the root modeling the initial string distribution; binary factors $w_{v,v'}(s, s') = \theta_{v'}(s,s')$ capturing string mutation probabilities between pairs of species connected by an edge $(v\to v')$ in the phylogenetic tree; and finally, unary factors $w_v(s) = \1[s = x_v]$ attached to each leaf $v\in\leavesset$, which are Dirac deltas on the observed strings.  

From this factor graph,  the likelihood can be written as:
\begin{align*}
\P_\tau(\observations) = 
\sum_{(s_v \in \Sigma^* : v\in\vertexset) \in (\Sigma^*)^{|\vertexset|}}
w_\Omega(s_\Omega) \left(  \prod_{v\in\leavesset} w_v(s_v) \right) \left( \prod_{(v\to v')\in \edgeset} w_{v,v'}(s_v, s_{v'})\right).
\end{align*}For any functions $f_i : {\mathcal X_i} \to \R^+$ over a countable spaces ${\mathcal X}$, we let $\prod_i f_i$ denote their pointwise product, 
and we let $\Vert f \Vert$ denote the sum of the function over the space, $$\Vert f \Vert = \sum_{x\in{\mathcal X}} f(x).$$The likelihood can then be expressed as follows:\begin{align}
\label{eq:graph-norm}
\P_\tau(\observations) = \left\Vert \prod_{w\in\potentials} w \right\Vert,
\end{align}So far, we have seen each potential $w$ as a black box, focussing instead on how they interact with each other.
In the next section, we describe the form assumed by each individual potential.

\subsection{Weighted automata and transducers}\label{sec:transd-fg}

Weighted automaton (respectively, transducers), is a classical way to define a function from the space of strings (respectively, the space of pairs of strings) to the non-negative reals \cite{Schutzenberger1961WA}.
At a high level, weighted automata and transducers proceed by extending a simpler weight function defined over a finite state space $\transspace$ into a function over the space of strings. This is done via a collection of paths (lists of variable length) in $\transspace$.

We explain how this is done, starting with automata (in the language of SFG, unary potentials). First, we designate one state in the state space $\transspace$ to be the \emph{start state} and one to be the \emph{end state} (we need only consider one of each without loss of generality). We define a \emph{path} as a sequence of states $q_i \in \transspace$ and \emph{emissions} $\hat \sigma_i \in \hat \Sigma = \Sigma \cup \{\epsilon\}$, starting at the start state and ending at the stop state: $p = (q_0,\hat \sigma_1, q_1, \hat\sigma_2, q_2,  \dots, \hat\sigma_n, q_n)$. Here the set of possible emissions $\hat \Sigma$ is the set of characters extended with the \emph{empty emission} $\epsilon$.\footnote{Note that we use the convention of using $\hat \sigma$ when the character is an element of the extended set of characters $\hat\Sigma$.}

Let us now assign a non-negative number to each triplet formed by a transition (pair of state $q, q' \in \transspace$) and emission $\hat \sigma \in \hat\Sigma$.\footnote{In technical terms, we use the Mealy model.} We call this map $w : \hat\Sigma \times \transspace^2 \to [0,\infty)$, a parameter specified by the user.\footnote{The inverse problem, learning the values for this map is also of interest \cite{Holmes2002}, but here we focus on the forward problem.} 

Next, we extend the input space of $w$ in two steps.
First, if $w$ takes a path $p$ as input, we define $w(p)$ as the product of the weights of the consecutive triplets $(q_i, \hat \sigma_{i+1}, q_{i+1})$ in $p$. 
Second, if $w$ takes a string $s \in \Sigma^*$ as input, we define $w(s)$ as the sum of the weights of the paths $p$ emitting $s$,
where a path $p$ emits a string $s$ if the list of characters in $p$, omitting $\epsilon$, is equal to $s$.

For transducers (binary potentials), we modify the above definition as follows. First, emissions become pairs of symbols in $\hat \Sigma \times \hat \Sigma$. Second, paths are similarly extended to emit pairs of symbols at each transition, $p = (q_0, (\hat\sigma_1,\hat\sigma'_1), q_1, (\hat\sigma_2, \hat\sigma'_2), q_2,  \dots, (\hat\sigma_n, \hat\sigma'_n), q_n)$. Finally a path emits a pair of strings $(s, s')$ if the concatenation of the first characters of the emissions, $\hat \sigma_1, \hat \sigma_2, \dots, \hat \sigma_n$, is equal to $s$ after removing the $\epsilon$ symbols, and similarly for $s'$ with the second characters of each emission, $\hat \sigma'_1, \hat \sigma'_2, \dots, \hat \sigma'_n$. For example, $p = (q_0, (\textrm{a},\epsilon), q_1, (\textrm{b}, \textrm{c}), q_2, (\epsilon, \epsilon), q_n)$ emits the pair of strings (ab, c).

The \emph{normalization} of a transducer or automata is the sum of all the valid paths' weights, $Z = \Vert w \Vert$.
We set the normalization  to positive infinity when the sum diverges.

We now have reviewed all the ingredients required to formalize the problem we are interested in:

\noindent{\bf Problem description: }Computing Equation~(\ref{eq:graph-norm}), where each $w\in\potentials$ is a weighted automaton or transducer organized into a tree-shaped SFG.

\section{Background}\label{sec:sum-prod-string-valued}

In the previous section, we have shown how the normalization of a single factor is defined (how to compute it in practice is a another issue, that we address in Section~\ref{sec:ops-on-fg}). In this section, we review how to transform the problem of computing the normalization of a full factor graph, Equation~(\ref{eq:graph-norm}), into the problem of finding the normalization of a single unary factor.
This is done using the elimination algorithm \cite{Bishop2006Book}, a classical method used to break complex computations on graphical models into a sequence of simpler ones.

The elimination algorithm is typically applied to find the normalization of factor graphs with finite-valued or Gaussian factors.  Since our factors take values in a different space, the space of strings, it is useful in this section to take a different perspective on the elimination algorithm: instead of viewing the algorithm as exchanging messages on the factor graph, we  view it as applying transformation on the topology of the graph.  These operations alter the normalization of individual potentials,  but, as we will show, they keep the normalization of the whole factor graph invariant, just as in standard applications of the elimination algorithm.

The elimination algorithm starts with the initial factor graph, and eliminates one leaf node at each step (either a variable node or a factor node, there must be at least one of the two types). This process creates a sequence of factor graphs with one less nodes at each step, until there are only two nodes left (one variable and one unary factor). While the individual factors are altered by this process, the global normalization of each intermediate factor graph remains unchanged.  

We show in Figure~\ref{fig:op-seqn} an example of this process, and a classification of the operations used to simplify the factor graph. We call the operations corresponding to a factor eliminations a \emph{pointwise product} (Figure~\ref{fig:op-seqn}(b,c)), and the operations corresponding to variable eliminations, \emph{marginalization} (Figure~\ref{fig:op-seqn}(a)).  
We further define two kinds of pointwise products: those where the variable connected to the eliminated factor is also connected to a binary factor (the first kind, shown in Figure~\ref{fig:op-seqn}(b)), and the others (second kind, shown in Figure~\ref{fig:op-seqn}(c)).

\begin{figure}[tp]
\begin{center}
  \includegraphics[width=4in]{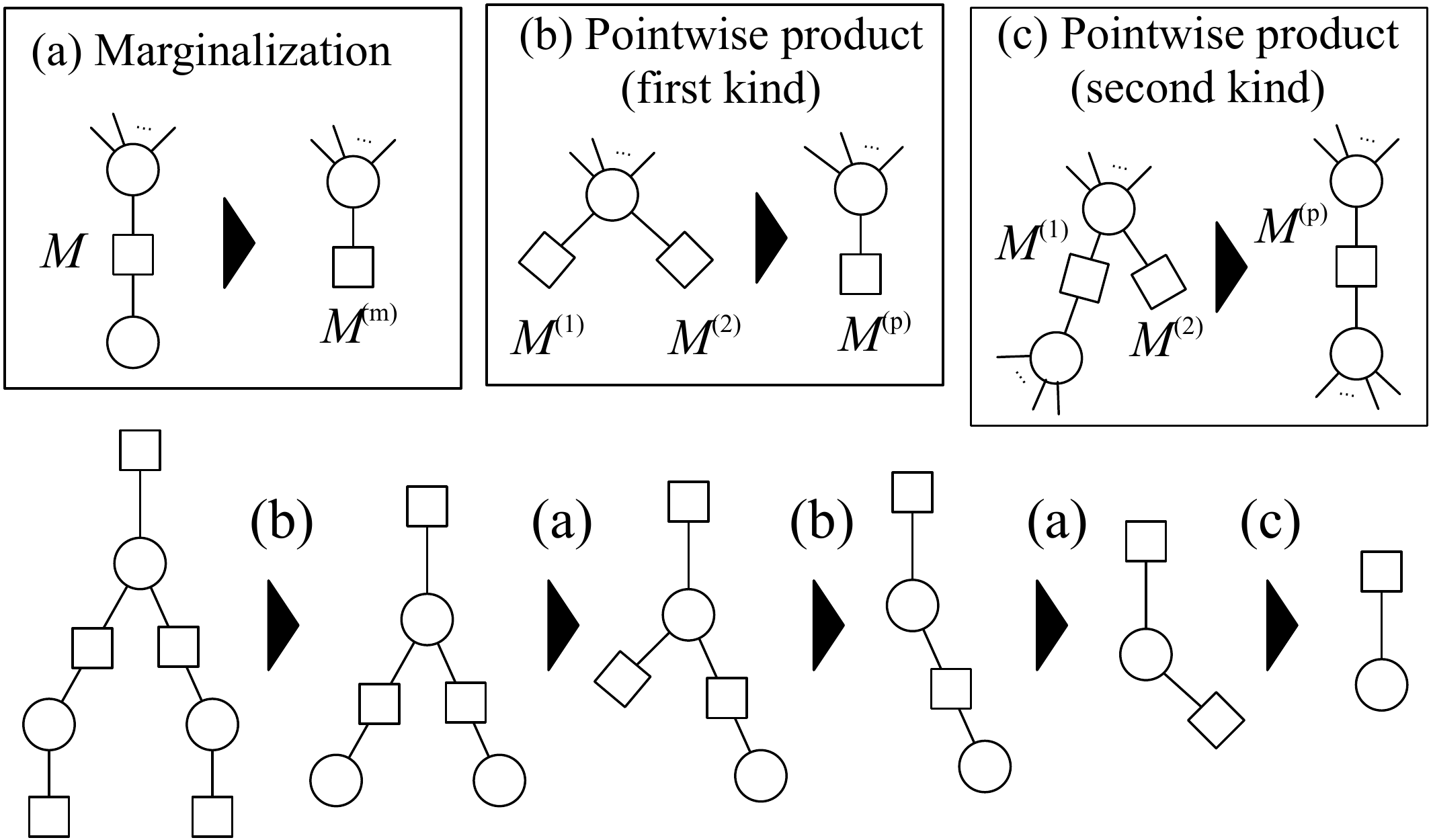}
  \caption[A sequence of fundamental probabilistic operations.]{Top: The tree graph transformations used in the elimination algorithm. Bottom: An example of the application of these three operations on a graph induced by a phylogenetic tree.}
  \label{fig:op-seqn}
\end{center}
\end{figure}

Let $\potentials$ denote a set of factors, and let $\potentials'$ denote a new set of factors obtained by applying one of the three operations described above.  For the elimination algorithm to be correct, the following invariant  should hold \cite{Bishop2006Book}:\begin{align*}
\left\Vert \prod_{w \in \potentials} w \right\Vert = \left\Vert \prod_{w' \in \potentials'} w' \right\Vert. 
\end{align*}

\section{Indexed matrices}
\label{sec:matrix-rep}

In this section, we describe the representation that forms the basis of our normalization method and establish the main properties of this representation.  We use TKF91 \cite{Thorne1991a} as a running example to illustrate  the construction.

\subsection{Unary factors}\label{sec:unary}
 
We assume without loss of generality that the states of the automata are labelled by the integers $1,\dots, K$, and that state $1$ is the  start state, and $K$, the  stop state.

We represent an automaton as a character-indexed collection, $\indexedmtx$, of $K\times K$ matrices, where $K$ is the size of the state space of the automaton, and the index runs over the extended characters: 
\begin{align*}
\indexedmtx = \bigg(M_{\hat\sigma}\in \M_K\left(\R^+\right)\bigg)_{\hat\sigma\in\hat\Sigma}
\end{align*}  
Each entry $M_{\hat \sigma}(k,k')$ encodes the weight of transitioning from states $k$ to $k'$ while emitting $\hat \sigma\in \hat \Sigma$.

An indexed collection $\indexedmtx$ specifies a corresponding weight function $w_\indexedmtx : \Sigma^* \to [0, \infty)$ as follows:  let $p$ denote a valid path, $p = q_0,\hat \sigma_1, q_1, \hat\sigma_2, q_2,  \dots, \hat\sigma_n, q_n$, and set
\begin{align*} 
w_\indexedmtx(p) = \phi\left(\prod_{i}^n M_{\hat\sigma_i}\right),
\end{align*}
where $\phi$  selects the entry $(1,K)$ corresponding to the product of weights starting from the initial and ending in the final state,  $\phi(M) = M(1,K)$.  The product denotes a matrix multiplication over the matrices indexed by the emissions in $p$, in the order of occurrence.  

In the TKF91 model for example, an automaton is needed to model the geometric root distribution, $\theta(s)$.  We have in Figure~\ref{fig:root-model-automaton} the classical state diagram as well as the indexed matrix representation, in the case where the length distribution has mean $1/(1-q-p)$, and the proportion of `a' symbols versus `b' symbols is $p/q$.  

\begin{figure}[tp]
\begin{center}
\begin{tabular}{ll}
\includegraphics[width=2.5in]{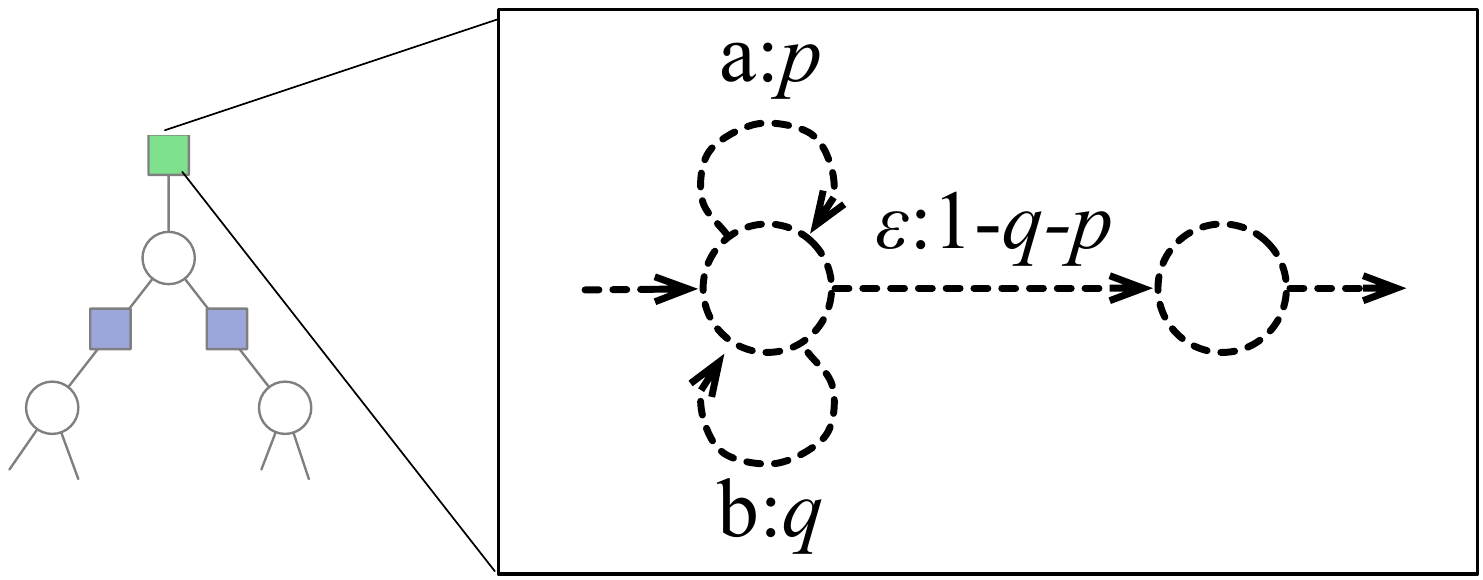} & 
\begin{minipage}{2in}
\vspace{-1in}
\begin{eqnarray*}
M_{\textrm{a}} =& \left(\begin{matrix} 
  p & \ \ 0 \\
  0 & \ \ 0
 \end{matrix}\right) \\
M_{\textrm{b}} =& \left(\begin{matrix}  
  q & \ \ 0 \\
  0 & \ \ 0
 \end{matrix}\right) \\
M_\epsilon =& \left(\begin{matrix} 
  0 & \ \ 1-p-q \\
  0 & \ \ 0
 \end{matrix}\right) \\
\end{eqnarray*}
\end{minipage}
\end{tabular}
  \caption[An initial string distribution automaton]{An initial geometric length distribution unary factor.   On the left,  we show the standard graph-based automaton representation, where states are dashed circles (to differentiate them from nodes in SFGs), and arcs are labelled by emissions and weights.  Zero weight arcs are omitted.  The start and stop state are indicated by inward and outward arrows, respectively on the first and second state here. On the right, we show the indexed matrix representation of the same unary potential. }
  \label{fig:root-model-automaton}
\end{center}
\end{figure}

Another important example of an automaton is one that gives a weight of one to a single observed string, and zero otherwise: $w_v(s) = \1[s = x_v]$.  We show the state diagram and the indexed matrix representation in Figure~\ref{fig:indicator-automaton}.  Note that the indexed matrices are triangular in this case, a property that will have important computational consequence in Section~\ref{sec:eff-analysis}.

\begin{figure}[tp]
% \begin{center}
\begin{tabular}{ll}
  \includegraphics[width=2.5in]{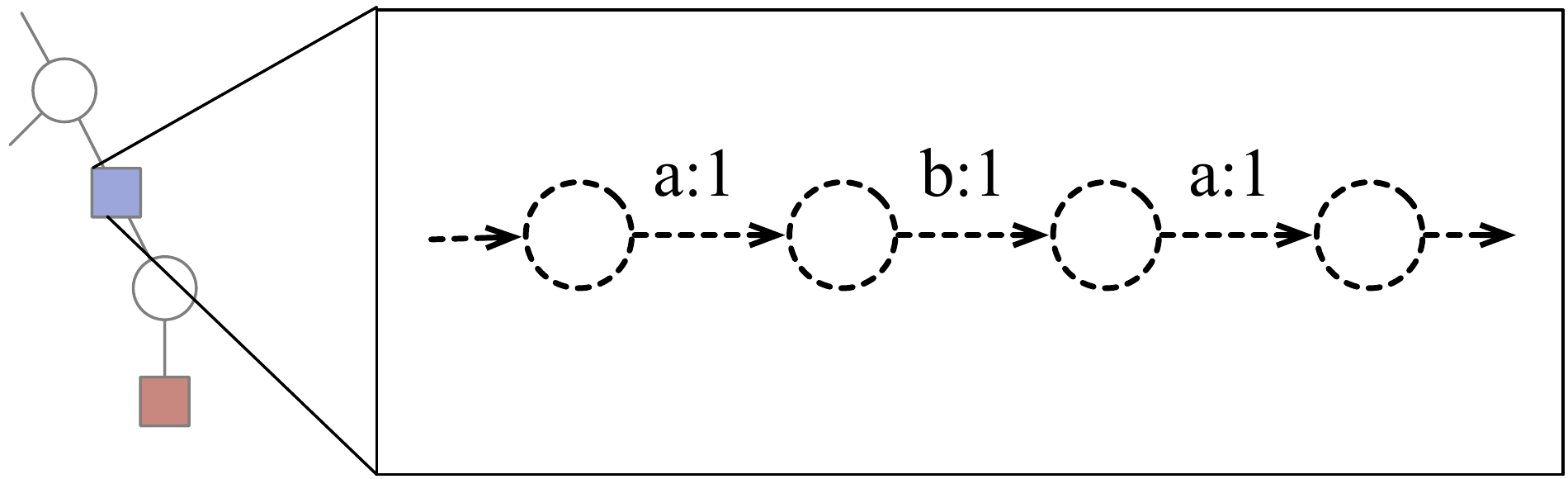} & 
\begin{minipage}{2in}
\vspace{-1in}
\begin{eqnarray*}
&M_{\textrm{a}} = \left(\begin{matrix} 
  0 & 1 & 0 & 0 \\
  0 & 0 & 0 & 0 \\
  0 & 0 & 0 & 1 \\
  0 & 0 & 0 & 0 \\ 
 \end{matrix}\right) \ \  
M_{\textrm{b}} = \left(\begin{matrix} 
  0 & 0 & 0 & 0 \\
  0 & 0 & 1 & 0 \\
  0 & 0 & 0 & 0 \\
  0 & 0 & 0 & 0 \\ 
 \end{matrix}\right) \ \ \\
&M_\epsilon~=~{\mathbf 0}
\end{eqnarray*}
\end{minipage}
\end{tabular}
  \caption[An automaton encoding a string indicator function.]{A unary factor encoding a string indicator function. In this example, only the string `aba' is accepted (i.e. the string `aba' is the only string emitted with positive weight).}
  \label{fig:indicator-automaton}
% \end{center}
\end{figure}

\subsection{Binary factors}\label{sec:binary-potentials}

We now turn to the weighted transducers, which we viewed as a collection of matrices as well, but this time with an index running over  pairs of symbols $(\hat \sigma,\hat\sigma')\in \hat \Sigma$: $$\indexedmtx = \bigg(M_{\hat\sigma,\hat\sigma'}\in \M_K\left(\R^+\right)\bigg)_{\hat\sigma,\hat\sigma'\in\hat\Sigma}$$  Note that in the context of transducers, epsilons need to be kept in the basic representation, otherwise transducers could not give positive weights to pairs of strings where the input has a different length than the output.  

In phylogenetics, transducers represent the key component of the probabilistic model: the probability of a string given its parent, $w_{v,v'}(s, s') = \theta_{v'}(s,s')$.  For example, in TKF91 models, this probability is determined by three parameters: an insertion rate $\lambda > 0$, a deletion rate $\mu > 0$, and a substitution rate matrix $Q$.  Given a branch of length $t$, the probability of all derivations between $s$ and $s'$ can be marginalized using the transducer shown in Figure~\ref{fig:transducer-tkf91}, where the values of $\alpha$, $\beta$ and $P$ are obtained by solving a system of differential equations yielding  \cite{Thorne1991a}, for all $\sigma,\sigma'\in\Sigma$:\begin{align*}
\alpha(t) &= \exp(-\mu t) \\
\beta(t) &= \frac{\lambda(1-\exp((\lambda - \mu)t))}{\mu - \lambda \exp((\lambda - \mu)t)} \\
\gamma(t) &= 1 - \frac{\mu\beta(t)}{\lambda(1-\alpha(t))} \\
P_t(\sigma,\sigma') &= \big(\exp(t Q)\big)(\sigma,\sigma'), 
\end{align*}and $\pi_{\sigma}$ is the stationary distribution of the process with rate matrix $Q$.  
While the standard description of TKF uses three states, our formalism allows us to express it with two states, with a state diagram shown in Figure~\ref{fig:transducer-tkf91}.  Note also that for simplicity we have left out the description of the `immortal link' \cite{Thorne1991a}, but it can be handled without increasing the number of states.  This is done by introducing an extra symbol $\# \notin \hat\Sigma$ flanking the end of the sequence. 

\begin{figure}[tp]
\begin{center}
\begin{tabular}{ll}
  \includegraphics[width=2.5in]{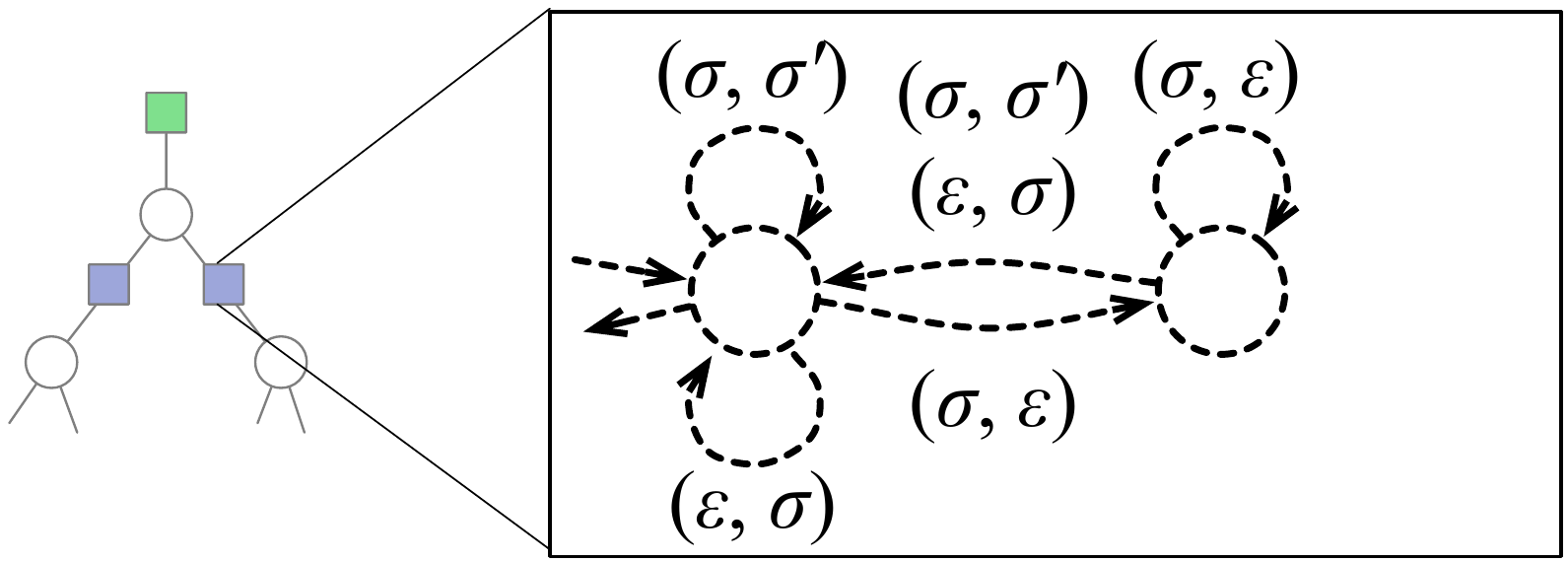} & 
\begin{minipage}{2in}
\vspace{-1in}
\begin{eqnarray*}
M_{\sigma,\sigma'} =& \frac{\alpha \lambda }{\mu}   \left(\begin{matrix}  % substitutions
  (1-\beta) & \ \ \ 0 \\   % MI
  (1-\gamma) & \ \ \ 0      % D
 \end{matrix}\right) P_{\sigma,\sigma'} \\
M_{\sigma,\epsilon}  =&  \frac{\lambda(1-\alpha)}{\mu} \left(\begin{matrix}   % deletion
  0 & \ \ \ (1-\beta)  \\   % MI
  0 & \ \ \ (1-\gamma)     % D
 \end{matrix}\right) \\
M_{\epsilon,\sigma}  =& \pi_{\sigma} \left(\begin{matrix}   % insertion
  \beta  & \ \ \ 0 \\   % MI
  \gamma  & \ \ \ 0      % D
 \end{matrix}\right)
\end{eqnarray*}
\end{minipage}
\end{tabular}
  \caption[An automaton encoding a string indicator function.]{A binary factor encoding a TKF91 marginal.}
  \label{fig:transducer-tkf91}
\end{center}
\end{figure}

\section{Operations on indexed matrices}\label{sec:ops-on-fg}

In this section, we give a closed-form expression for all the operations described in Section~\ref{sec:background}.  We start with the simplest operations, epsilon-removal and normalization, and then cover marginalization and pointwise products.

\subsection{Unary operations}

While epsilon symbols are convenient when defining new automata, some of the algorithms in the next section assume that there are no epsilon transitions of positive weights.  Formally, a factor $M$ is called \emph{epsilon-free} if $M_\epsilon = 0$.  Fortunately, converting an arbitrary unary factor  into an epsilon-free factor, can be done using the well known generalization of the geometric series formula to matrices. But in order to apply this result to our situation, we need to assume that in all positive weight path, the last emission is never an epsilon. Fortunately, this can be done without loss of generality by adding a boundary symbol at the end of each string. We assume this thorough the rest of this paper and get:

\begin{proposition}\label{prop:epsilon-removal}
Let $M$ denote a unary factor 
such that each of the eigenvalues $\lambda_i$ of $M_\epsilon$ satisfies $|\lambda_i|< 1$.
Then the transformed factor $M^\efindex$ defined as:
\begin{align*} 
M^\efindex_{\hat\sigma} &= \bracearraycond{0&\textrm{if }\hat\sigma = \epsilon\\M_\epsilon^* M_{\hat\sigma}&\textrm{otherwise}} \\
M^* &= (1 - M)^{-1}
\end{align*}is epsilon-free and assigns the same weights to strings, i.e. $w_{\indexedmtx}(s) = w_{\indexedmtx^\efindex}(s)$ for all $s\in\Sigma^*$.
\end{proposition}

We show a proof of this elementary result to illustrate the conciseness of our indexed matrices notation:

\begin{proof} 
Note first that the condition on the eigenvalues implies that $(1 - M)$ is invertible.

To prove equivalence of the weight functions $w_\indexedmtx$ and $w_{\indexedmtx^\efindex}$, let $s\in\Sigma^*$ be a string of length $N$.
We need to show that the epsilon-removed automaton $M^\efindex_{\sigma}$ assigns the same weights $w_{M^\efindex}(s)$ to strings as the original automaton $M_{\hat\sigma}$: %We have:
\begin{align*} 
w_\indexedmtx(s) &= \phi\left(\sum_{\hat s \leadsto s} \prod_{\hat \sigma \in \hat s} M_{\hat\sigma}\right) \\
&=  \phi\left(\sum_{(k_1,\dots,k_N)\in\N^N} M_\epsilon^{k_1} M_{s_1} M_\epsilon^{k_2} M_{s_2} \cdots M_\epsilon^{k_N} M_{s_1}\right) \\
&= \phi\left(\prod_{n=1}^N \left( \sum_{m=0}^\infty M_\epsilon^m  \right) M_{s_n}\right) \\
&= \phi\left(\prod_{n=1}^N M^\efindex_{s_n}\right) \\
&= w_{M^\efindex}(s)
\end{align*}

Here we used the nonnegativity of the weights, which implies that if the automaton has a finite normalization, then the infinite sums above are absolutely convergent, and can therefore be rearranged.
\end{proof}

We also need the following elementary result on the normalizer of arbitrary automata.  

\begin{proposition}\label{prop:normalization}
Let $M$ 
denote a unary potential, and define the matrix:
\begin{align*}
\tilde M = \sum_{\hat\sigma\in\hat\Sigma} M_{\hat\sigma},
\end{align*}with eigenvalues $\lambda_i$.
Then the  normalizer $Z_\indexedmtx$ of the unary potential is given by:
 \begin{align*} 
Z_\indexedmtx = 
\begin{cases}
\phi\left(\tilde M^*\ \right) & \textrm{if for all }i, |\lambda_i| < 1 \\
+\infty & \textrm{otherwise}
\end{cases}
\end{align*}where $\phi(M) = M(1,K)$.
\end{proposition}

\begin{proof} 
Note first that for all indexed matrices $\indexedmtx$, we have the following identity:
\begin{align*} 
\sum_{(\hat\sigma_1, \dots, \hat\sigma_N)\in \hat \Sigma^L} \prod_{n=1}^N M_{\sigma_n} = \left(\sum_{\sigma\in\hat\Sigma} M_{\hat\sigma}\right)^N 
= \tilde M^N.
\end{align*}
Therefore by decomposing the set of all paths by their path lengths $N$, we obtain:
\begin{align*} 
Z_\indexedmtx &= \phi\left( \sum_{N=0}^\infty \sum_{(\hat\sigma_1, \dots, \hat\sigma_N)\in \hat \Sigma^L} \prod_{n=1}^N M_{\sigma_n} \right) \\
&= \phi\left( \sum_{N=0}^\infty \tilde M^N \right)
\end{align*}
The infinite sum in the last line converges to $M^*$ if and only if the eigenvectors satisfy $|\lambda_i| < 1$, and diverges to $+\infty$ otherwise since the weights are non-negative.

\end{proof}

\subsection{Binary operations}

We now describe how the operations of marginalization and pointwise product can be efficiently implemented using our framework.  We begin  with the marginalization operation.

The marginalization operation that we first focus on, shown in Figure~\ref{fig:op-seqn}, takes as input a binary factor $\indexedmtx$, and returns a unary factor $\indexedmtx^\mindex$.  A precondition of this operation is that one of the nodes connected to $\indexedmtx$ should have no other potentials attached to it.  We use the variable $\sigma'$ for the values of that node. This node is eliminated from the factor graph after applying the operation.  The other node is allowed to be attached to other factors, and its values are denoted by $\sigma$.

\begin{proposition}\label{prop:marginalization}
If $M$ is a binary factor, then the unary factor $M^\mindex$ defined as:
\begin{align*} 
M^\mindex_{\hat \sigma} &= \sum_{\hat \sigma' \in \hat \Sigma} M_{\hat \sigma, \hat \sigma'}
\end{align*}satisfies 
$w_{M^\mindex}(s) = \sum_{s'\in\Sigma^*} w_M(s,s')$ 
for all $s\in\Sigma^*$.
\end{proposition}

\begin{proof} 
We first note that for any matrices $A_{\hat t}, \hat t\in\hat\Sigma^*$ with non-negative entries, we have
\begin{align*} 
\sum_{t\in \Sigma^*} \sum_{N=0}^\infty \sum_{\hat t \leadsto_N t} A_{\hat t} = \sum_{N=0}^{\infty} \sum_{\hat t\in \hat \Sigma^N} A_{\hat t},
\end{align*}where we write $\hat s \leadsto_L s$ if $\hat s$ has length $L$ and $\hat s \leadsto s$.

We now fix $s\in \Sigma$.  Applying the result above, with
\begin{align*} 
A_{\hat s'} = \sum_{\hat s \leadsto_N s}  \prod_{n=1}^N M_{\hat s_n, \hat s'_n},
\end{align*}for all $\hat s'$, and where $N = |\hat s'|$, we get:
\begin{align*} 
\sum_{s'\in \Sigma^*} w_M(s, s') &= \phi\left( \sum_{s'\in\Sigma^*} \sum_{N=0}^\infty \sum_{\hat s \leadsto_N s}  \sum_{\hat s' \leadsto_N s'} \prod_{n=1}^N M_{\hat s_n, \hat s'_n} \right) \\
&= \phi\left( \sum_{s'\in\Sigma^*} \sum_{N=0}^\infty   \sum_{\hat s' \leadsto_N s'} A_{\hat s'} \right) \\
&=  \phi\left( \sum_{N=0}^{\infty} \sum_{\hat s'\in \hat \Sigma^N} A_{\hat s'} \right) \\
&= \phi\left(\sum_{N=0}^\infty \sum_{\hat s \leadsto_N s} \prod_{n=1}^N \sum_{\hat \sigma \in \hat\Sigma} M_{\hat s_n, \hat\sigma} \right) \\
&= w_{M^\mindex}(s)
\end{align*}
\end{proof}

Next, we move to the pointwise multiplication operations, where the epsilon transitions need a special treatment. 
In the following, we use the notation $\{A|B|C| \dots\}$ to denote the tensor product of the matrices $A$, $B$, $C$, \dots.\footnote{Note that we avoided the standard tensor product notation $\otimes$  because of a notation conflict with the automaton and transducer literature, in which $\otimes$ denotes multiplication in an abstract semi-ring (the generalization of normal multiplication, $\cdot$ used here).  The operator $\otimes$ is also often overloaded to mean the product or concatenation of automata or transducers, which is not the same as the \emph{pointwise} product as defined here.}

\begin{proposition}\label{prop:pointwise1}
If $M^{(1)}$ and $M^{(2)}$ are two epsilon-free unary factors, then the unary factor $M^\pindex$ defined as:
\begin{align}\label{eq:tensor-prod}  
M^\pindex_{\alpha} = \left\{ M^{(1)}_\sigma \bigg| M^{(2)}_\sigma \right\}
\end{align}is epsilon-free and satisfies:
\begin{align}\label{eq:pointwise-conclusion}
w_{M^\findex}(s) \cdot w_{M^\sindex}(s) = w_{M^\pindex}(s)
\end{align}for all $s\in\Sigma^*$.
\end{proposition}

\begin{figure}[tp]
\begin{center}
\begin{tabular}{ll}
  \includegraphics[width=1in]{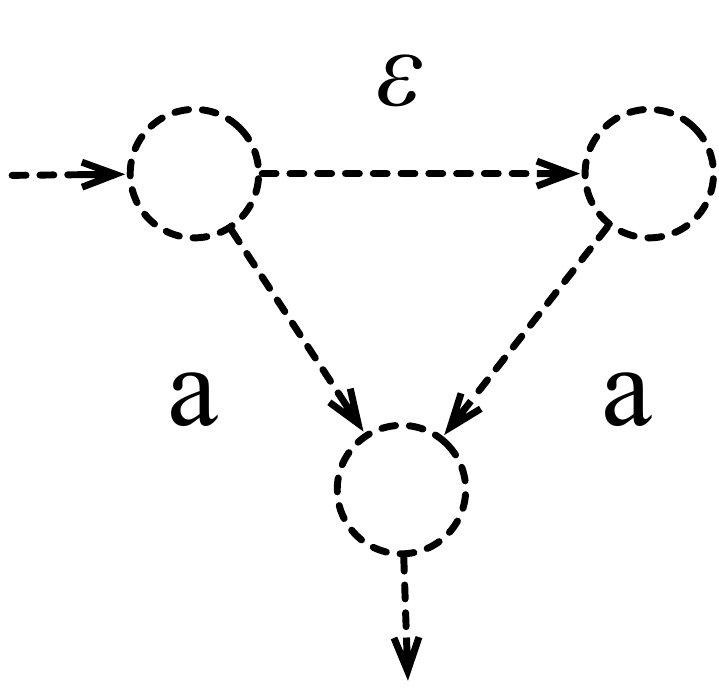} & 
\begin{minipage}{2in}
\vspace{-1in}
\begin{eqnarray*}
M_{\textrm{a}} =& \left(\begin{matrix}  % substitutions
  0 & 0 & \alpha \\   % MI
  0 & 0 & \alpha \\
  0 & 0 & 0
 \end{matrix}\right) \\
M_{\epsilon}  =&  \left(\begin{matrix}   % deletion
  0 & \alpha & 0 \\
  0 & 0 & 0 \\
  0 & 0 & 0
 \end{matrix}\right) 
\end{eqnarray*}
\end{minipage}
\end{tabular}
  \caption[An automaton encoding a string indicator function.]{Counter-example used to illustrate that the epsilon-free condition in Proposition~\ref{prop:pointwise1} is necessary.  Here $\alpha$ is an arbitrary constant in the open interval $(0, 1)$.}
  \label{fig:counter-eg}
\end{center}
\end{figure}

\begin{proof}
Clearly, $M^\pindex_\epsilon$ is equal to the zero matrix, so $M^\pindex$ is epsilon-free.  To show that it satisfies Equation~\ref{eq:pointwise-conclusion}, we first note that since $M^\iindex$ is epsilon free, 
\begin{align*} 
\sum_{\hat s \leadsto s} \prod_{\hat \sigma\in s} M^\iindex_{\hat\sigma} = \prod_{\sigma \in s} M^\iindex_\sigma,
\end{align*}for 
$i\in\{1,2\}$. Hence:\begin{align*} 
w_{M^\findex}(s) \cdot w_{M^\sindex}(s) &= \phi\left( \sum_{\hat s \leadsto s} \prod_{\hat \sigma\in s} M^{(1)}_{\hat\sigma} \right) \phi\left( \sum_{\hat s' \leadsto s} \prod_{\hat \sigma'\in s'} M^{(1)}_{\hat\sigma'} \right) \\
&= \phi\left(\prod_{\sigma \in s} M^{(1)}_\sigma\right)  \phi\left(\prod_{\sigma \in s} M^{(2)}_\sigma\right) \\
&= \phi\left(\prod_{\sigma\in s} \left\{ M^{(1)}_\sigma \bigg| M^{(2)}_\sigma \right\}\right) \\
&= w_{M^\pindex}(s),
\end{align*}for all $s\in\Sigma^*$.
\end{proof}

Note that the epsilon-free condition is necessary.  Consider for example the unary potential $M$ over $\Sigma = \{\textrm{a}\}$ shown in Figure~\ref{fig:counter-eg}.  We claim that this factor $M$ (which has positive epsilon transitions) does not satisfy the conclusion of Proposition~\ref{prop:pointwise1}.  Indeed, for $s = \textrm{a}$, we have that $(w_M(s))^2 = \alpha^4 + 2 \alpha^3 + \alpha^2$, but $w_{M^\pindex}(s) = \alpha^4 + \alpha^2$, where $M^\pindex$ is defined as in Equation~\ref{eq:tensor-prod}.  

We now turn to pointwise products of the second kind, where an automaton is pointwise multiplied with a transducer.  The difference is that all epsilons cannot be removed from a transducer: without emissions of the form $(\sigma,\epsilon), (\epsilon, \sigma)$, transducers would not have the capacity to model insertions and deletions.  Fortunately, as long as the automaton $M^{(2)}$ is epsilon free, pointwise multiplications of the second kind can be implemented as follows:

\begin{proposition}\label{prop:pointwise2}
If $M^\sindex$ is epsilon free and $M^\findex$ is any transducer, then the transducer $M^\pindex$ defined by:
\begin{align*} 
M^\pindex_{\hat\sigma, \hat\sigma'} &= \bracearraycond{\left\{M_{\hat\sigma, \hat\sigma'}^\findex\big| M_{\hat\sigma}^\sindex\right\}&\textrm{if } \hat\sigma\neq\epsilon\\
\left\{M_{\hat\sigma, \hat\sigma'}^\findex\big| I\right\}&\textrm{otherwise}}
\end{align*}satisfies 
\begin{align*} 
w_{M^\findex}(s,s') \cdot w_{M^\sindex}(s) &= w_{M^\pindex}(s,s')
\end{align*}for all $s, s' \in \Sigma^*$.
\end{proposition}

\begin{proof} 
Note first that since $M^\sindex$ is epsilon-free, we have
\begin{align*} 
\prod_{\alpha\in s} M^\sindex_\alpha = \prod_{n = 1}^N \bracearraycond{I&\textrm{if }\hat s_n = \epsilon\\M^\sindex_{\hat s_n}&\textrm{otherwise}}
\end{align*}for all $N \ge |s|$, $\hat s \leadsto_N s$.  Moreover, for all matrix $A_i$, $\phi(\sum_i A_i) = \sum_i \phi(A_i)$, hence:
\begin{align*} 
w^\findex(s,s') \cdot w^\sindex(s) &= \phi\left(\sum_{N=0}^\infty \sum_{\hat s \leadsto_N s} \sum_{\hat s' \leadsto_N s} \prod_{n=1}^N M^\findex_{\hat s_n,\hat s'_n}\right) \phi\left(\prod_{\sigma \in s} M^\sindex_\sigma\right) \\
&= \sum_{N=0}^\infty \sum_{\hat s \leadsto_N s} \sum_{\hat s' \leadsto_N s} \phi\left( \prod_{n=1}^N M^\findex_{\hat s_n,\hat s'_n} \right)  \phi\left( \prod_{n=1}^N \bracearraycond{I&\textrm{if }\hat s_n = \epsilon\\M^\sindex_{\hat s_n}&\textrm{otherwise}} \right) \\
&= \sum_{N=0}^\infty \sum_{\hat s \leadsto_N s} \sum_{\hat s' \leadsto_N s} \phi\left( \prod_{n=1}^N M^\pindex_{\hat s_n, \hat s'_n} \right) \\
&= w_{M^\pindex}(s,s')
\end{align*}
\end{proof}

Now that we have a simple expression for all of the operations required by the elimination algorithm on string-valued graphical models, we turn to the problem of analyzing the time complexity of exact inference in SFGs.

\section{Complexity analysis}\label{sec:eff-analysis}

In this section, we start by demonstrating the gains in efficiency brought by our framework in the general case, that is without assuming any structure on the potentials.  We then show that by adding one extra assumption on the factors, triangularity, we get further efficiency gains.  We conclude the section by giving several examples where triangularity holds in practice.

\subsection{General SFGs}\label{sec:general-smrfs}

Here we compare the local running-time of the indexed matrix representation with previous approaches based on graph algorithms \cite{Mohri2009}. By local, we mean that we do the analysis for a single operation at the time.  We present global complexity analyses for some important special cases in the next section. The results here are stated in terms of $K$, the size of an individual transducer, but note that in sequence alignment problems, $K$ grows in terms of $L$ because of potentials at the leaves of the tree. This is explained in more details in the next section.

\begin{description}
  \item[Normalization:] Previous work has relied on generalizations of shortest path algorithms to Conway semirings \cite{Droste2009Semirings}, which run in time $O(K^3)$.  In contrast, since our algorithm is expressed in terms of a single matrix inversion, we achieve a running time of $O(K^\omega)$, where $\omega < 2.3727$ is the exponent of the asymptotic complexity of matrix multiplication \cite{Williams2012MatrixMult}. Note that even though the size of the original factors' state spaces may be small (for example of the order of the length of the observed sequences), the size of the intermediate factors created by the tensor products used in the elimination algorithm makes $K$ grow quickly.  
  \item[Epsilon-removal:] Classical algorithms are also cubic \cite{Mohri2002EpsilonRemoval}.  Since we perform epsilon removal using one matrix multiplication followed by one matrix inversion, we achieve the same speedup of $O(K^{3-\omega})$.  Note however that this operation does not preserve sparsity in the general case, both with the classical and the proposed method.  This issue is addressed in the next section. 
  \item[Pointwise products:] The operation corresponding to pointwise products in the previous transducer literature is the \emph{intersection} operation \cite{Eilenberg1974Automata,Mohri2009}.  
In this case, the complexity of our algorithm is the same as previous work, but our algorithm is easier to implement when one has access to a matrix library. 
Moreover, tensor products preserve sparsity.  More precisely, we use the following simple results: if $A$ has $k$ nonzero entries, and $B$, $l$ nonzero entries, then forming the tensor product takes time $kl$.    We assume throughout this section that matrices are stored in a sparse representation. 
\end{description}

Note that these running times scale linearly in $|\Sigma|$ for operations on unary factors,  and quadratically in $|\Sigma|$ for operations on binary factors.  

\subsection{Triangular SFGs}\label{sec:triangular-smrfs}

When executing the elimination algorithm on SFGs, the intermediate factors produced by pointwise multiplications are fed into the next operations, augmenting the size of the matrices that need to be stored.  We show in this section that by making one additional assumption, part of this growth can be managed using sparsity and properties of tensor products.

We will cover the following two SFG topologies: the star SFG, and the perfect binary SFG, shown in Figure~\ref{fig:first-step}.  We denote the number of leaves by $L$, and we let $N$ and $c$ be bounds on the sizes of the unary potentials at the leaves and the binary potentials respectively (by a size $N$ potential $M$, we mean that the matrices $M_\sigma$ have size $N$ by $N$).  For example, in the phylogenetics setup, $N$ is  equal to an upper bound on the lengths of the sequences being analyzed, $L$ is the number of taxa under study, and $c=2$ in the case of a TKF91 model.
We start our discussion with the star SFG

In Figure~\ref{fig:first-step}, we show the result of the first two steps of the elimination algorithm.\footnote{We omit the factor at the root since its effect on the running time is a constant independent of $L$ and $N$.}  The operations involved in this first step are pointwise products of the second kind. 
The intermediate factors produced, $M^{(1)}, \dots, M^{(L)}$ are each of size $cN$.  The second step is a marginalization, which does not increase the size of the matrices.   We therefore have:\begin{align}\label{eq:slow-eq}
Z = \phi\left( \sum_{\hat\sigma\in\hat\Sigma} \left\{ \left(M^{(1)}_\epsilon\right)^* M^{(1)}_{\hat\sigma} \bigg| \dots \bigg| \left(M^{(L)}_\epsilon\right)^* M^{(L)}_{\hat\sigma} \right\} \right)^*
\end{align}

Computing the right-hand-side naively would involve inverting an $(cN)^L$ by $(cN)^L$ dense matrix---even if the matrices $M^{(l)}_{\hat\sigma}$ are sparse, there is no a priori guarantee for $\left(M^{(l)}_\epsilon\right)^*$ to be sparse as well.  This would lead to a slow running time of $(cN)^{L\omega}$. In the rest of this section, we show that this can be considerably improved by  using properties of tensor products and an extra assumptions on the factors:

\begin{definition}\label{def:triangular}
A factor, transducer or automaton is called \emph{triangular} if its states can be ordered in such a way that all of its indexed matrices are upper triangular.  A SFG is triangular if all of its leaf potentials are triangular.
\end{definition}

Note that we do not require that all the factors of SFG be triangular.  In the case of TKF91 for example, only the factors at the leaves are triangular.  This is enough to prove the main proposition of this section:

\begin{proposition}\label{prop:star-running-time}
If a star-shaped or perfect binary triangular SFG has $L$ unary potentials of size $N$ and binary potentials of size $c$ then the normalizer of the SFG can be computed in time $(cN)^L$.
\end{proposition}

\begin{figure}[tp]
\begin{center}
  \includegraphics[width=4in]{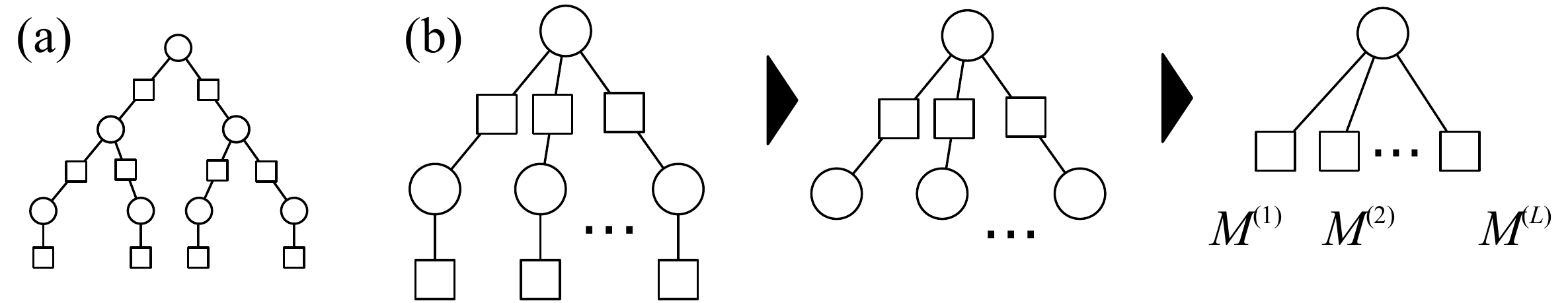}
  \caption[The first two steps of the elimination algorithm on the star tree.]{(a) The binary topology.  (b) The star topology with the first two steps of the elimination algorithm on that topology.  Note that in both cases we ignore the root factor, which only affect the running time by a constant independent of $L$ and $N$.}
  \label{fig:first-step}
\end{center}
\end{figure}

In order to prove this result, we need the following standard properties \cite{Langville2004SAN}:

\begin{lemma}\label{lemma:mtx1} 
Let $A^{(l)}$ be $m\times k$ matrices, and $B^{(l)}$ be $k\times n$ matrices.  We have: \begin{align*} 
\left\{A^{(1)} B^{(1)} \bigg| \dots \bigg| A^{(L)} B^{(L)} \right\} &= \left\{ A^{(1)} \bigg| \dots \bigg| A^{(L)} \right\} \left\{ B^{(1)} \bigg| \dots \bigg| B^{(L)} \right\} 
\end{align*}
\end{lemma}

\begin{lemma}\label{lemma:mtx2}
If $A^{(l)}$ are invertible, then:\begin{align*} 
\left\{ \left(A^{(1)}\right)^{-1} \bigg| \dots \bigg| \left(A^{(L)}\right)^{-1} \right\} =  \left\{ A^{(1)} \bigg| \dots \bigg| A^{(L)} \right\}^{-1}
\end{align*}
\end{lemma}

\begin{lemma}\label{lemma:mtx3} 
If $A$ is a $m\times k$ invertible matrix, and $B$ is a $k\times n$ matrix such that $A-B$ is invertible, then $\left(A^{-1} B\right)^* = A (A-B)^{-1}$.
\end{lemma}

The next lemma states that in a sense triangular factors are contagious:

\begin{lemma}\label{lemma:contagion}
If a factor is triangular, then marginalizing it or making it epsilon-free creates a triangular potential.  Moreover, the outcome of a pointwise product (either of the first or second kind) is guaranteed to be triangular whenever at least one of its input automaton or transducer is triangular.
\end{lemma}

\begin{proof} 
For marginalization and pointwise products, the result follows directly from Proposition~\ref{prop:marginalization}, \ref{prop:pointwise1}, and \ref{prop:pointwise2}.   For epsilon-removal, it follows from Proposition~\ref{prop:epsilon-removal} and the fact that the inverse of a triangular matrix is also triangular.  
\end{proof}

We can now prove Proposition~\ref{prop:star-running-time}:

\begin{proof}
Note first that by Lemma~\ref{lemma:contagion}, the matrices $M^{(i)}_\epsilon$ is upper triangular.
Since the intermediate factors represent probabilities, $w_{M^{(l)}}(s) = \P_\tau(X_\Omega = s, \observations)$,  we have that $M^{(i)}_\epsilon(j,j) < 1$.  It follows that the eigenvalues of $M^{(i)}_\epsilon$ are strictly smaller than one, so Proposition~\ref{prop:epsilon-removal} can be applied to $M^{(i)}$ for each $i\in\{1, \dots, L\}$.

Next, by applying Proposition~\ref{prop:normalization}, we get:
\begin{align*}
Z &= \phi\left( \sum_{\hat\sigma\in\hat\Sigma} \left\{ \left(M^{(1)}_\epsilon\right)^* M^{(1)}_{\hat\sigma}\bigg| \dots\bigg| \left(M^{(L)}_\epsilon\right)^* M^{(L)}_{\hat\sigma} \right\} \right)^*
\end{align*}This expression can be simplified using Lemma~\ref{lemma:mtx1}
\ref{lemma:mtx2}, and \ref{lemma:mtx3} to get:
\begin{align*}
Z &= \phi\left( \left\{ \left(M^{(1)}_\epsilon\right)^* \bigg| \dots\bigg| \left(M^{(L)}_\epsilon\right)^* \right\}  \sum_{\hat\sigma\in\Sigma} \left\{  M^{(1)}_{\hat\sigma}\bigg| \dots\bigg|  M^{(L)}_{\hat\sigma} \right\} \right)^* \\
&= \phi\left( \left\{ \bar M^{(1)}_\epsilon \bigg| \dots\bigg| \bar M^{(L)}_\epsilon \right\} \left( \left\{ \bar M^{(1)}_\epsilon \bigg| \dots\bigg| \bar M^{(L)}_\epsilon \right\} - \sum_{\hat\sigma\in\Sigma} \left\{  M^{(1)}_{\hat\sigma}\bigg| \dots\bigg|  M^{(L)}_{\hat\sigma} \right\} \right)^{-1} \ \right) \\
&= \phi\left( A (A-B)^{-1} \right),
\end{align*}where $\bar M = I - M$,  $A = \left\{ \bar M^{(1)}_\epsilon \bigg| \dots\bigg| \bar M^{(L)}_\epsilon \right\}$, and $B = \sum_{\hat\sigma\in\Sigma} \left\{  M^{(1)}_{\hat\sigma}\bigg| \dots\bigg|  M^{(L)}_{\hat\sigma} \right\}$.

This equation allows us to exploit sparsity patterns.  By Lemma~\ref{lemma:contagion}, the matrices $ M^{(i)}_\epsilon$, $M^{(i)}_{\hat\sigma}$  are upper triangular sparse matrices (more precisely: $cN$ by $cN$ matrices with $cN$ non zero entries), so the matrices $A$ and $B$ are upper triangular as well, with only $O((cN)^L)$ non-zero components.  Furthermore, if we let $C = (A-B)^{-1}$, we can see that only its last column ${\mathbf x}$ is actually needed to compute $Z$.  At the same time, we can write $(A-B) {\mathbf x} = [0, 0, \dots, 0, 1]$, which can be solved using the back substitution algorithm, getting an overall running time of $O((cN)^L)$.

Note that in the above argument, the only assumption we have used on the factors at the leaves is triangularity.  The same argument hence applies when these factors come from the elimination algorithm executed on a binary tree.  The running time for the binary tree case can therefore be established by an inductive argument on the depth where the inductive step is the same as the argument above.
\end{proof}

\section{Discussion}\label{sec:discussion}

By exploiting closed-form matrix expressions rather than specialized automaton and transducer algorithms, our method gains several advantages. Our algorithms are arguably easier to understand and implement since their building blocks are well-known linear algebra operations. At the same time, our algorithms have asymptotic running times better or equal to previous transducer algorithms, so there is no loss of performance incurred, and even potential gains. In particular, by building implementations on top of existing matrix packages, our method automatically gets access to  optimized libraries \cite{Whaley2001Atlas}, or to GPU parallelization from off-the-shelf libraries such as gpumatrix \cite{Mingming2012GPU}. 
Last but not least, our method also has applications as a proving tool.   For example, despite a large body of theoretical work on the complexity of inference in TKF91 models \cite{Hein2001Generalization,Steel2001TKF,Hein2001TKF,Lunter2003TKF,Hein2003RecursionsPNAS}, basic questions such as computing the expectation of additive metrics  under TKF91 remain open \cite{Daskalakis2011AlignFree}. Our concise closed-form expression for exact inference could be useful to attack this type of problems. 
 
In phylogenies that are too large to be handled by exact inference algorithms, our results can still be used as building blocks for approximate inference algorithms.  For example, the acceptance ratio of most existing MCMC samplers for TKF91 models can be expressed in terms of normalization problems on a small subtree of the original phylogeny \cite{Jensen2002}.  The simplest way to obtain this subtree is to fix the value of the strings at all internal nodes except for two adjacent internal nodes.  Resampling a new topology for these two nodes amounts to computing the normalization of two  SFGs over a subtree of four leaves (their ratio appears in the Metropolis-Hastings ratio) \cite{Redelings2005}. 
Our approach can also serve as the foundation of Sequential Monte Carlo (SMC) approximations \cite{teh2007,gorur2008,Bouchard2011SMC}.  In this case, the ratio of normalizations appears as a particle weight update. The matrix formulation also opens the possibility of using low-rank matrices as an approximation scheme, a direction that we leave to future work.

\newpage 

{\footnotesize
\bibliographystyle{plain}
\bibliography{references}
}

\end{document}